\documentclass[final]{siamltex}
\usepackage{amsfonts,color,pslatex}
\usepackage{amssymb,latexsym,amsmath}
\newcommand{\gf}{{\mathbb {F}}}

\newcommand{\fr}{\gf_r}

\newcommand{\ftwom}{\gf_{2^m}}

\newcommand{\tr}{\mathrm {Tr}}

\newcommand{\ubar}{\overline{u}}

\newcommand{\exponent}{{\rm exp}}

\title{Permutation trinomials over finite fields with even characteristic\thanks{
The research of C. Ding was supported by the Hong Kong Research Grants Council, Project No. 601013. 
The research of L. Qu was supported by the Natural Science Foundation of China (No. 61272484)
and the Basic Research Fund of the National University of Defense Technology (No. CJ 13-02-01). 
The research of Q. Wang was supported by the NSERC of Canada and  National Natural Science Foundation of China (No. 61170289). 
The research of P. Yuan was supported by the NSF of China (Grant No. 11271142) and
the Guangdong Provincial Natural Science Foundation (Grant No. S2012010009942). 
}
}

\author{
Cunsheng Ding\thanks{Department of Computer Science and Engineering,
        Hong Kong University of Science and Technology, Clear Water Bay,
        Kowloon, Hong Kong ({\tt cding@ust.hk}).} 
\and Longjiang Qu\thanks{College of Science, National University of Defense Technology, Changsha, Hunan, P. R. China, 410073 (\tt ljqu\_happy@hotmail.com).} 
\and Qiang Wang\thanks{School of Mathematics and Statistics,  
Carleton University,  
1125 Colonel By Drive, Ottawa, Ontario,  
K1S 5B6, Canada (\tt wang@math.carleton.ca).} 
\and Jin Yuan\thanks{3/16 Vimiera Rd, Eastwood, NSW 2122, Australia (\tt jin.yuan.au@gmail.com).}   
\and Pingzhi Yuan \thanks{School of Mathematics, South China Normal University, Guangzhou 510631, China  (\tt  yuanpz@scnu.edu.cn).}    
}

\begin{document}

\maketitle

%\maketitle 

\begin{abstract}
Permutation polynomials have been a subject of study for a long time and have applications in many areas 
of science and engineering. However, only a small number of specific classes of permutation polynomials are 
described in the literature so far. In this paper we present a number of permutation trinomials over finite fields, 
which are of different forms. 

\end{abstract}



\maketitle 

\section{Introduction} 

A permutation polynomial $f(x)$ over a finite field is a polynomial that acts as a permutation of the elements of the field, 
i.e. the map $x \mapsto f(x)$ is one-to-one. Permutation polynomials are closely related to highly nonlinear functions  
\cite{DY06,Dobb992} and other areas of mathematics.  They have applications in combinatorial 
designs \cite{DY06,Dobb992}, coding theory \cite{CK03,Yann,ST05} and cryptography \cite{MN85,QTTL,RSA}. For instance, Dickson 
permutation polynomials of order five, i.e.,  
$D_5(x, a)=x^5+ax^3-a^2x$ over $\gf_{3^m}$,  led to a 70-year research breakthrough in 
combinatorics \cite{DY06}, gave a family of perfect nonlinear functions for cryptography 
\cite{DY06}, generated good linear codes \cite{CDY} for data communication and 
storage, and produced optimal signal sets for CDMA communications \cite{DingYin}, to 
mention only a few applications of these Dickson permutation polynomials.      
Information about properties, 
constructions, and applications of permutation polynomials can be found in Cohen \cite{Cohn}, Lidl and Niederreiter 
\cite{LN}, and Mullen \cite{Mull}. Some recent progress on permutation polynomials can be found in 
\cite{AGW11,CaoHuZha14,FH12,Hou,LCX13,LHT13,MullenWang13,Wang13,WuLiu13,YuanDing11,ZZH10,ZhaHu12}.

Permutation polynomials with fewer terms over finite fields with even characteristics  are in particular interesting. For example, in the study of Almost perfect nonlinear (APN) mappings which are of interest for their applications in cryptography,  Dobbertin first proved
a well-known conjecture of Welch stating that for odd $n=2m+1$, the power function $x^{2m+3}$ is even maximally nonlinear \cite{Dobb99a} or, in other terms, that the crosscorrelation function between a binary maximum-length linear shift register sequence of degree $n$ and a decimation of that sequence by $2m+3$ takes on precisely the three values $-1$, $-1 \pm 2^{m+1}$. 
The method in fact relies on the discovery of a class of permutation trinomials. Around the same time period,  Dobbertin proved Niho's conjecture similarly using a class of permutation pentanomials \cite{Dobb99b}. Another example of demonstrating the application of permutation polynomials with fewer terms in constructing cyclic codes can be found in \cite{Ding131}. We note that these interesting permutation polynomials are of the simple form, i.e., they have nonzero coefficients equal to the identity. In this case,  the permutation monomial is trivial to find and there is no permutation binomials with both nonzero coefficients $1$ over finite fields with even characteristic.  This motivates us to search for more classes of permutation trinomials with nonzero trivial coefficients, which are potentially useful in above mentioned applications. However, only a small number of classes of permutation trinomials over $\ftwom$ are known in the literature. To the best of the authors' 
knowledge, the following is a list of known 
classes of permutation trinomials over $\ftwom$: 
\begin{enumerate}
\item Some linearized permutation trinomials described in \cite{LN}.

\item $x+x^3+x^5$ over $\ftwom$, where $m$ is odd (the Dickson
polynomial of degree 5).

\item $x+x^5+x^7$ over $\ftwom$, where $m \not\equiv 0  \pmod{3}$ (the Dickson
polynomial of degree 7). 

\item $x+x^3+x^{2^{(m+1)/2}+1}$ over $\ftwom$, where $m$ is odd \cite{Dobb99a}.  

\item $x^{2^{2k}} +(ax)^{2^k+1} + ax^2$, where $m=3k$ and $a^{(2^m-1)/(2^k-1)} \ne 1$ \cite{BCHO}. 

\item $x^{3 \times 2^{(m+1)/2}+4}+x^{2^{(m+1)/2}+2}+x^{2^{(m+1)/2}}$,  
where $m$ is odd (\cite{Cher98} or \cite[Theorem 4]{Dobb02}). 
\end{enumerate}

In this paper, we present a few new classes of permutation trinomials over $\ftwom$  through the study of the number of solutions of special equations. In Section 2, we  deal with  permutation trinomials  over $\ftwom$ such that $m$ is odd. In contrast, we obtain a few more classes of permutation trinomials  over $\ftwom$ such that $m$ is even in Section 3. We point out some potential applications in the Section 4 and hope that the interested readers will find the usage of these polynomials in constructing linear codes, bent functions, 
and difference sets. Throughout this paper, $\tr_m(x)$ denotes the absolute trace function on $\gf_{2^m}$.

%\section{Preliminaries}  
  
%Throughout this paper, $\tr_m(x)$ denotes the absolute trace function on $\gf_{2^m}$. 
%We will need the following lemma in the sequel. 

%\begin{lemma}\label{lem-evenquadratic} \cite{LN} 
%Let $m$ be a positive integer. The equation $x^2+ux+v=0$, where $u,v \in \gf_{2^m}$, $u\neq 0$, has roots in $\gf_{2^m}$ if and only if $\tr_{m}(v/u^2)=0$.
%\end{lemma}  

\section{The case that $m$ is odd} 

%The following theorem is easy to prove. 

%\begin{theorem} 
%For any odd integer $m>1$, $f(x)=x + x^{(2^m +1)/3} + x^{(2^{m+1} -1)/3}$ is a permutation 
%polynomial over $\gf_{2^m}$. 
%\end{theorem}  

%\begin{proof} Note that $x^3$ and and $g(x)=(1+x)^3$ are permutation polynomials over $\ftwom$. Since $f(x^3)=g(x)$, $f(x)$ is a permutation polynomials. This completes the proof.  \end{proof}

%\begin{theorem} 
%For any odd integer $m>1$, 
%$f(x)=x+x^{2^{(m+1)/2}+1} +x^{2^{(m+3)/2}+3}$ is a permutation 
%polynomial over $\gf_{2^m}$.  
%\end{theorem}  

%\begin{proof} 
%It is known that $h(x)=x^{3\cdot 2^{(m+1)/2}+4}+x^{2^{(m+1)/2}+2}+x^{2^{(m+1)/2}}$  
%is a permutation polynomial (in fact it is an \emph{o}-polynomial, see \cite{Cher98} or
%\cite[Theorem 4]{Dobb02}). Note that $h(y^{2^{(m-1)/2}})=f(y)=y+y^{2^{(m+1)/2}+1} +y^{2^{(m+3)/2}+3}$ 
%and that $y^{2^{(m-1)/2}}$ is a permutation over $\gf_{2^m}$.  Hence, $f(x)$ is a permutation trinomial over 
%$\gf_{2^m}$. This completes the proof.  \end{proof}

%The following theorem follows from the fact that $(x^{-1}+1)^3$ is a permutation 
%polynomial over $\ftwom$. 

%\begin{theorem}\label{thm-4second3}
%Let $m$ be an odd positive integer. Then
%$$
%f(x)=x^{2^{m-2}-1} + x^{2^{m-1}-1} + x^{2^m - 2^{m-2}-1}
%$$
%is a permutation polynomial of $\ftwom$.
%\end{theorem} 

The first family of permutation trinomials are given in the following theorem.  

\begin{theorem}\label{thm-june28}
For any odd integer $m>1$, $f(x)=x + x^{2^{(m+1)/2}-1} + x^{2^m-2^{(m+1)/2}+1}$ is a permutation 
polynomial over $\gf_{2^m}$. 
\end{theorem}  

\begin{proof} 
We have that 
\begin{eqnarray*}
f(x) &=& x + x^{2^{(m+1)/2}-1} + x^{2^m-2^{(m+1)/2}+1} \\
     &=& x\left(1 + x^{2^{(m+1)/2}-2} + x^{2^m-2^{(m+1)/2}} \right) \\
     &=& x \left(1 + x^{2\cdot ( 2^{(m-1)/2}-1)} + x^{2^{(m+1)/2} (2^{(m-1)/2}-1)} \right).
\end{eqnarray*}
 Because $\gcd(2^{(m-1)/2} -1, 2^m -1) =1$, $f(x)$ is a permutation polynomial over $\gf_{2^m}$ if and only if 
$$g(x) = x^{2^m -1 - (2^{(m+1)/2} +2)} \left(1+x^2 + x^{2^{(m+1)/2}} \right)$$ 
is a permutation polynomial over $\gf_{2^m}$. We note that $g(0) =0$ and 
$$g(x) = \frac{1+x^2 + x^{2^{(m+1)/2}}}{x^{2^{(m+1)/2} +2}}$$ 
when $x \neq 0$.

First of all, we show that $x=0$ is the only solution to $g(x) =0$. If $g(x) =0$, then either $x=0$ or $1+x^2 +  x^{2^{(m+1)/2}} = 0$. If $1+x^2 +  x^{2^{(m+1)/2}} = 0$, then   $1+x^{2^{(m+1)/2}} + x^{2^m} =0$ after we raise both sides to the power of $2^{(m-1)/2}$. Adding the above two equations, we obtain $x^{2^m} + x^2 =0$.  Therefore we obtain either $x =0$ or $x = 1$. However, $g(1) =1$. Hence the only solution to $g(x) =0$ is $x=0$.

Next we prove  that $g(x) =a$ has a  unique nonzero solution for each nonzero $a \in \gf_{2^m}$. That is, for each nonzero $a\in \gf_{2^m}$, we prove that there exists a unique nonzero solution $x$ for the following equation:
\begin{equation} 
\label{eqn1}
 \frac{1+x^2 + x^{2^{(m+1)/2}}}{x^{2^{(m+1)/2} +2}} = a. 
\end{equation}
Rewriting Equation~(\ref{eqn1}), we obtain the following equation
\begin{equation} 
\label{eqn2}
 a x^{2^{(m+1)/2} +2} + x^{2^{(m+1)/2}} +  x^2 + 1 =0.
\end{equation}
Let $y=x^2$. Then Equation~(\ref{eqn2}) becomes that  
\begin{equation} 
\label{eqn3}
 a y^{2^{(m-1)/2} +1} + y^{2^{(m-1)/2}} +  y + 1 =0. 
\end{equation}

Now we try to solve Equation~(\ref{eqn3}) for each nonzero $a$ and nonzero $y$. 
First, if $a=1$, then  we have 
$$ y^{2^{(m-1)/2} +1} + y^{2^{(m-1)/2}} +  y + 1 = (y^{2^{(m-1)/2}} +1)(y+1) =(y+1)^{2^{(m-1)/2} +1}=0.$$ 
 Hence  $y=1$ is the unique nonzero solution to Equation~(\ref{eqn3}) for $a=1$.

From now on, we assume $a \neq 1, 0$.
Raising  the power of $2^{(m+1)/2}$ to Equation~(\ref{eqn3}), we obtain
\begin{equation} 
\label{eqn4}
 a^{2^{(m+1)/2}} y^{2^{(m+1)/2} +1} + y + y^{2^{(m+1)/2}}  + 1 =0.
\end{equation}

Adding Equations~(\ref{eqn3}) and ~(\ref{eqn4}), we have
\begin{equation}
\label{eqn5}
 a^{2^{(m+1)/2}} y^{2^{(m+1)/2} +1} + ay^{2{(m-1)/2}+1} + y^{2^{(m+1)/2}}  + y^{2^{(m-1)/2}} =0.
\end{equation}

Because $y\neq 0$, dividing Equation~(\ref{eqn5}) by $y^{2^{(m-1)/2}}$ results in
\begin{equation}
\label{eqn6}
 a^{2^{(m+1)/2}} y^{2^{(m-1)/2} +1} + ay + y^{2^{(m-1)/2}}  + 1 =0.
\end{equation}

Therefore the sum of Equations~(\ref{eqn3}) and (\ref{eqn6}) yields
\begin{equation}
\label{eqn7}
 (a^{2^{(m+1)/2}} +a) y^{2^{(m-1)/2} +1} + (a+1)y  =0.
\end{equation}

Because $y\neq 0$, we obtain $(a^{2^{(m+1)/2}} +a) y^{2^{(m-1)/2} } + (a+1)  =0$. Since $(a^{2^{(m+1)/2}} +a) \neq 0$ for $a\neq 0, 1$, the polynomial $(a^{2^{(m+1)/2}} +a) y^{2^{(m-1)/2}} + (a+1)$ is a permutation polynomial over $\gf_{2^m}$ as $\gcd(2^{(m-1)/2}, 2^m -1) =1$. Hence there exists a unique nonzero solution $y$ to  $(a^{2^{(m+1)/2}} +a) y^{2^{(m-1)/2}} + (a+1)  =0$ for $a\neq 0, 1$. Hence there exists at most one nonzero solution $x$ to Equation~(\ref{eqn1}) for each nonzero $a$. Therefore there exists a unique solution to $g(x) =a$ for each $a$. Hence the proof is complete. 
\end{proof}

The following theorem follows from Theorem \ref{thm-june28}. 

\begin{theorem}\label{thm-4second5}
Let $m>1$ be an odd positive integer. Then
$$
f(x)=x^{2^{(m-1)/2}-1} + x^{2^m-2^{(m-1)/2}-2} + x^{2^m - 2^{(m-1)/2} -1}
$$
is a permutation polynomial of $\ftwom$.
\end{theorem} 

The third family of permutation trinomials is described in the next theorem. 

\begin{theorem}\label{thm-oddmain}
Let $m$ be an odd integer. Then
$$
f(x)=x+x^3+x^{2^m-2^{(m+3) / 2}+2}
$$
is a permutation polynomial of $\ftwom$.
\end{theorem}

\begin{proof} 
Let $d=(m+1)/2$ and  $y=x^{2^d}$. Then $y^{2^d}=x^2$ and for any $x\neq 0$,
$$f(x)=x+x^3+\frac{x^3}{y^2}=\frac{x(x^2y^2+y^2+x^2)}{y^2}.$$

We firstly prove that $f(x)=0$ if and only if $x=0$. Assume, on the contrary, there
exists some $x\in\gf_{2^m}^*$ such that
\begin{equation}\label{EqQ1}
  x^2y^2+y^2+x^2=0.
\end{equation}
Raising the above equation to its $2^d$-th power, we obtain  
\begin{equation}\label{EqQ2}
  x^4y^2+x^4+y^2=0.
\end{equation}
We compute \eqref{EqQ1}$^2+$ \eqref{EqQ2} as follows:
$$(x^4+1)(y^4+y^2)=0.$$
Hence we get $x=y=1$. It is a contradiction since  $f(1)=1\neq 0$.
Thus $f(x)=0$ if and only if $x=0$.

If $f(x)$ is not a permutation, then there exists $x\in \gf_{2^m}^*$
and $a\in \gf_{2^m}^*$ such that $f(x)=f(x+ax)$. Let $b=a^{2^d}$. It is clear that $a, b\neq 0,1$.
Then
$$\frac{x(x^2y^2+y^2+x^2)}{y^2}=\frac{(a+1)x((a^2+1)(b^2+1)x^2y^2+(b^2+1)y^2+(a^2+1)x^2)}{(b^2+1)y^2}.$$
After simplifying and rearranging the terms, we have  
\begin{equation}\label{EqGen}
  A_1x^2y^2+A_2y^2+A_3x^2=0,
\end{equation}
where
\begin{eqnarray*}
  A_1 &=& a(b+1)^2(a^2+a+1), \\
  A_2 &=& a(b+1)^2, \\
  A_3 &=& (a+1)^3+(b+1)^2.
\end{eqnarray*}
Raising \eqref{EqGen} to its $2^d$-th power, we obtain  
\begin{equation}\label{EqGen2}
  A_1^{2^d}x^4y^2+A_3^{2^d}y^2 + A_2^{2^d}x^4=0.
\end{equation}
We compute \eqref{EqGen}$*(A_1^{2^d}x^4+A_3^{2^d})+$\eqref{EqGen2}$*(A_1x^2+A_2)$ to cancel $y^2$:
\begin{equation}\label{EqGen3}
  B_1x^4+B_2x^2 + B_3=0.
\end{equation}
where
\begin{eqnarray*}
  B_1 &=& A_3A_1^{2^d}+A_1A_2^{2^d}=(a+1)^4b^2[(a+1)^3+(b+1)^3], \\
  B_2 &=& A_2^{2^d+1}=ab(a+1)^4(b+1)^2, \\
 B_3 &=& A_3^{2^d+1}.
\end{eqnarray*}
Now we claim that $A_1, A_2, B_1, B_2\neq 0$. It can be easily proved and is
left to the interested readers.

Let $x^2=\frac{B_2}{B_1}\gamma$. Plugging it into \eqref{EqGen3}, we get  
\begin{equation}\label{EqGam}
  \gamma^2+\gamma+D=0,
\end{equation}
where $D=\frac{B_1B_3}{B_2^2}$. Furthermore, we have
\begin{equation*}
  D=\frac{B_1B_3}{B_2^2}=\frac{A_3^{2^d+1}(A_3A_1^{2^d}+A_1A_2^{2^d})}{A_2^{2^{d+1}+2}}=\frac{A_1A_3^{2^d+1}}{A_2^{2^d+2}}
  +\frac{A_1^{2^d}A_3^{2^d+2}}{A_2^{2^{d+1}+2}}=D_1+D_1^{2^d},
\end{equation*}
where $D_1=\frac{A_1A_3^{2^d+1}}{A_2^{2^d+2}}=\frac{A_1B_3}{A_2B_2}$.

Now we have the following claim: 

{\bf Claim 1. $\tr_m(D_1)=1$ for any $a\in \gf_{2^m}\setminus\{\gf_2\}$. }

Claim 1 will be proved later.

Raising \eqref{EqGam} to the $2^i$-th power, where $i=0,1,\cdots, d-1$ and then summing them up,
we have
\begin{equation*}
  \gamma^{2^d}=\gamma+\sum_{i=0}^{d-1}(D_1+D_1^{2^d})^{2^i}=\gamma+\sum_{i=0}^{2d-1}D_1^{2^i}=\gamma+D_1+\tr_m(D_1)=\gamma+D_1+1.
\end{equation*}
and
\begin{equation*}
  \gamma^{2^d+1}=\gamma(\gamma+D_1+1)=D_1\gamma+D.
\end{equation*}
Plugging the two equations above into \eqref{EqGen},  we obtain  
$$\frac{A_1B_2^{2^d+1}}{B_1^{2^d+1}}(D_1\gamma+D)+\frac{A_2B_2^{2^d}}{B_1^{2^d}}(\gamma+D_1+1)+\frac{A_3B_2}{B_1}\gamma=0.$$
Multiplying $B_1^{2^d+1}$ across the two sides of the above equation and using $B_2^{2^d}=A_2^{2^d+2}=A_2B_2$, we have 
\begin{equation}\label{EqGam2}
  C_1\gamma+C_2=0,
\end{equation}
where \begin{eqnarray*}
  C_1 &=& A_1A_2B_2D_1+A_2^2B_1+A_3B_1^{2^d}, \\
  C_2 &=& A_1A_2B_2D+A_2^2B_1(D_1+1).
\end{eqnarray*}
Substituting $D, D_1, B_1$ into the above two equations, and after simplifications, we can easily verify that
$C_1=A_1A_2^{2^d+2}$ and $C_2=A_2^2B_1$. Hence by \eqref{EqGam2}, we have
\begin{equation*}
  \gamma=\frac{C_2}{C_1}=\frac{B_1}{A_1A_2^{2^d}}=\frac{A_2B_1}{A_1B_2}.
\end{equation*}
Plugging it into \eqref{EqGam}, and recalling that $D=\frac{B_1B_3}{B_2^2}$,  we get 
\begin{equation*}
  B_1A_2^2+A_1A_2B_2=A_1^2B_3.
\end{equation*}
Using the definitions of $B_1, B_2$ and $B_3$ to simplify the above equation, we deduce  that 
\begin{equation*}
  A_1^{2^d}A_2^2=A_1^2A_3^{2^d}.
\end{equation*}
Hence we have
\begin{equation*}
  (a+1)^4[(b+1)^3+1]a^2(b+1)^4=(b+1)^4[(a+1)^6+1][(b+1)^3+(a+1)^4],
\end{equation*}
which is reduced to
\begin{equation*}
  (b+1)^3=(a+1)^6.
\end{equation*}
Since $\gcd(3, 2^m-1)=1$, we have $b+1=a^2+1$, ie.
$ 
  b=a^2.
$ 
Raising it to the $2^d$-th power, we get  
\begin{equation*}
  a^2=b^2=b.
\end{equation*}
It is a contradiction since $b\neq 0,1$.

Now,  to finish the proof of the theorem, it suffices to prove Claim 1.

Let $$R=\frac{A_3}{A_2}+1=\frac{(a+1)(a+b)^2}{a(b+1)^2}.$$ Then
$$D_1=\frac{A_1A_3^{2^d+1}}{A_2^{2^d+2}} = \frac{A_1}{A_2}\cdot (\frac{A_3}{A_2})^{2^d+1}=(a^2+a+1)(R+1)^{2^d+1}.$$
Hence
\begin{eqnarray*}
  \tr_m(D_1) &=& \tr_m((a^2+a+1)R^{2^d+1})+\tr_m((a^2+a)R^{2^d})+\\             
                    &  & \tr_m((a^2+a)R)+\tr_m(R^{2^d}+R)+\tr_m(a^2+a+1). 
\end{eqnarray*}
Let $E_1=\tr_m((a^2+a+1)R^{2^d+1})$, $E_2=\tr_m((a^2+a)R^{2^d})$ and $E_3=\tr_m((a^2+a)R)$. Then
\begin{equation}\label{EqD1}
  \tr_m(D_1)=E_1+E_2+E_3+1.
\end{equation}
Firstly,
\begin{equation*}
  E_3=\tr_m((a^2+a)R)=\tr_m\left(\frac{(a+1)^2(a+b)^2}{(b+1)^2}\right),
\end{equation*}
\begin{equation*}
  E_2=\tr_m((a^2+a)R^{2^d})=\tr_m((b^2+b)R^2)=\tr_m\left(\frac{(a+1)^2(b^2+b)(a+b)^4}{a^2(b+1)^4}\right),
\end{equation*}
and \begin{eqnarray*}
      E_1 &=& \tr_m((a^2+a+1)R^{2^d+1}) \\
      &=&\tr_m\left((a^2+a+1)\cdot \frac{(a+1)(a+b)^2}{a(b+1)^2} \cdot \frac{(b+1)(a^4+b^2)}{b(a+1)^4} \right) \\
       &=& \tr_m\left(\frac{(a^3+1)(a+b)^2(a^4+b^2)}{a(a+1)^4(b^2+b)} \right) \\
       &=&  \tr_m\left(\frac{(b^3+1)(a^4+b^2)(a^4+b^4)}{b(b+1)^4(a^4+a^2)} \right).
    \end{eqnarray*}
Noting in the second equality of $E_2$ and the last one of $E_1$, we raise left hands of equalities to the $2^d$-th power.
Hence

\begin{eqnarray*}
      E_1+E_2 &=& \tr_m\left(\frac{(a^4+b^4)}{a^2b(b+1)^4(a^2+1)}\left[(b^3+1)(a^4+b^2)+(a+1)^4(b^2+b) b\right] \right) \\
      &=& \tr_m\left(\frac{(a^4+b^4)}{a^2b(b+1)^4(a^2+1)}\left[(a^4+b^3)(b+1)^2\right] \right) \\
      &=& \tr_m\left(\frac{(a^4+1+b^4+1)(a^4+b^3)}{a^2b(b+1)^2(a^2+1)} \right) \\
      &=& \tr_m\left(\frac{(a+1)^2(a^4+b^3)}{a^2b(b+1)^2}\right)+\tr_m\left( \frac{(b+1)^2(a^4+b^3)}{a^2b(a+1)^2} \right) \\
      &=& F_1+F_2,
    \end{eqnarray*}
where $F_1=\tr_m\left(\frac{(a+1)^2(a^4+b^3)}{a^2b(b+1)^2}\right)$ and $F_2=\tr_m\left( \frac{(b+1)^2(a^4+b^3)}{a^2b(a+1)^2} \right)$.
We have
\begin{equation*}
  F_2=\tr_m\left( \frac{(b+1)^2(a^4+b^3)}{a^2b(a+1)^2}\right)=\tr_m\left( \frac{(a+1)^4(a^6+b^4)}{a^2b^2(b+1)^2}\right)=\tr_m\left( \frac{(a+1)^2(a^3+b^2)}{ab(b+1)}\right),
\end{equation*}
and
\begin{eqnarray*}
F_1 + E_3 &=& \tr_m\left(\frac{(a+1)^2}{a^2b(b+1)^2}\left[(a^4+b^3)+a^2b(a+b)^2\right]\right) \\
&=& \tr_m\left(\frac{(a+1)^2\left[a^4(b+1)+(a+1)^2b^3\right]}{a^2b(b+1)^2}\right) \\
&=& \tr_m\left(\frac{a^2(a+1)^2}{b(b+1)}\right) +  \tr_m\left(\frac{(a+1)^2b}{a(b+1)}\right).
\end{eqnarray*}
Thus
\begin{eqnarray*}
E_1 + E_2+E_3 &=& F_1+F_2+E_3\\
&= &\tr_m\left(\frac{1}{ab(b+1)}\left[a^3(a+1)^2 +(a+1)^2b^2+ (a+1)^2(a^3+b^2)\right]\right)\\
&=& 0.
\end{eqnarray*}
Then Claim 1 follows from \eqref{EqD1}. The proof is now completed. 
\end{proof}

The following theorem follows from Theorem \ref{thm-oddmain} easily. 

\begin{theorem}\label{thm-4second4}
Let $m$ be an odd positive integer. Then
$$
f(x)=x^{2^{m-2}-1} + x^{2^{m-2} + 2^{(m-1)/2}-1} + x^{2^m - 2^{m-2}-1}
$$
is a permutation polynomial of $\ftwom$.
\end{theorem}

\section{The case that $m$ is even}

We will need the following lemma in the sequel. 

\begin{lemma}\label{lem-evenquadratic} \cite{LN} 
Let $m$ be a positive integer. The equation $x^2+ux+v=0$, where $u,v \in \gf_{2^m}$, $u\neq 0$, has roots in $\gf_{2^m}$ if and only if $\tr_{m}(v/u^2)=0$.
\end{lemma}  

Suppose $r=2^{m}$ with $m$ even. For any $u\in\fr$, we use $\ubar$ to denote $u^{2^{m/2}}$.
Clearly, we have $\overline{\overline{u}}=u$ for any $u \in \gf_{2^m}$.

\begin{theorem} 
For any even integer $m \ge 2$, $f(x)=x + x^{2^{(m+2)/2}-1} + x^{2^m-2^{m/2}+1}$ is a permutation 
polynomial over $\gf_{2^m}$. 
\end{theorem}  

\begin{proof} 
Let $m_1=m/2$ and $r_1=2^{m/2}$. Then $r=r_1^2$ and  
$$ 
f(x)=x+x^{2r_1-1}+x^{2-r_1}. 
$$
For any $a \in \gf_{r}$,  we want to prove that $f(x)=a$ has a unique solution $x \in \gf_{r}$. Let $f(x)=a$. 
Then we have 
\begin{eqnarray}\label{eqn-canahk}
x+x^{2r_1-1}+x^{2-r_1}=a. 
\end{eqnarray} 

Raising both sides of (\ref{eqn-canahk}) to the power of $r_1$ yields 
\begin{eqnarray}\label{eqn-canahk2}
x^{r_1}+x^{2r_1-1}+x^{2-r_1}=a^{r_1}. 
\end{eqnarray} 
Combining (\ref{eqn-canahk}) and (\ref{eqn-canahk2}) gives 
\begin{eqnarray}\label{eqn-canahk3}
x^{r_1}+x=a^{r_1}+a. 
\end{eqnarray} 
Hence $(x+a)^{r_1}=x+a$. This means that $x+a \in \gf_{r_1}$. 

We first consider the case that $a \in \gf_{r_1}$. In this case, $a+a^{r_1}=0$. It then follows from (\ref{eqn-canahk2}) 
that $x^{r_1}=x$. Then the equation $f(x)=a$ becomes that $x=a$. Hence, $x=a$ is the unique solution in this case.  

We then deal with the case that $a \in \gf_{r} \setminus \gf_{r_1}$. In this case the minimal polynomial of $a$ over  
$\gf_{r_1}$ must be of the form 
$$ 
x^2+ux+v=0, 
$$
where $u \in \gf_{r_1}$ and $v \in \gf_{r_1}$. Since $a$ and $a^{r_1}$ are all the distinct roots of $x^2+ux+v=0$, we have 
$$ 
u=a+a^{r_1},  \ \ v=a^{1+r_1}. 
$$ 
In this case, $x^2+ux+v$ must be irreducible over $\gf_{r_1}$. By Lemma \ref{lem-evenquadratic}, we have that 
\begin{eqnarray}\label{eqn-tom}
\tr_{m_1}\left( \frac{v}{u^2} \right)=1.   
\end{eqnarray} 

In this case $x=0$ is not a solution of $f(x)=a$. Multiplying $x^{r_1+1}$ on both sides of (\ref{eqn-canahk}), we 
obtain that 
\begin{eqnarray}\label{eqn-canahk4}
x^{2+r_1}+x^{3r_1}+x^{3}=ax^{1+r_1}. 
\end{eqnarray} 
Combining (\ref{eqn-canahk3}) and (\ref{eqn-canahk4}) yields 
\begin{eqnarray}\label{eqn-canahk5}
x^3 + ax^2 + (a+a^{r_1})a^{r_1}x + (a+a^{r_1})^3=0. 
\end{eqnarray}   
Putting $x=y+a$ in (\ref{eqn-canahk5}), we obtain 
\begin{eqnarray}\label{eqn-root30}
y^3+by+c=0, 
\end{eqnarray} 
where $b=a^2+a^{1+r_1} + a^{2r_1} \in \gf_{r_1}$ and $c=(a+a^{r_1})b \in \gf_{r_1}$.  
If $b=0$, then $c=y=0$ and $x=a$. Hence it has a unique solution. In the following 
we assume that $b\neq 0$.
Hence $c \ne 0$. 
Our task in this case is to prove that  (\ref{eqn-root30}) has at most one solution $y \in \gf_{r_1}$.

On the contrary, suppose that (\ref{eqn-root30}) has two distinct solutions $y_1$ and $y_2$ in $\gf_{r_1}$. Then we have 
\begin{eqnarray}\label{eqn-root31}
y_1^3+by_1+c=0
\end{eqnarray}  
and 
\begin{eqnarray}\label{eqn-root32}
y_2^3+by_2+c=0. 
\end{eqnarray} 
Since $b \ne 0$ and $c \ne 0$, $y_i \ne 0$ for $i=1$ and $i=2$.

Subtracting  (\ref{eqn-root32})  from (\ref{eqn-root31}), we obtain 
\begin{eqnarray}\label{eqn-root24}
y_1^2+y_1y_2+y_2^2+b=0. 
\end{eqnarray}  
Let $z=y_1y_2^{-1}$. It then follows from  (\ref{eqn-root24}) that 
\begin{eqnarray}\label{eqn-root35}
z^2+z+1+\frac{b}{y_2^2}=0. 
\end{eqnarray}  
Hence 
\begin{eqnarray}\label{eqn-root36}
\tr_{m_1}\left(1+\frac{b}{y_2^2}\right )=\tr_{m_1} \left(z + z^2\right )=0. 
\end{eqnarray}  

Note that $u=a+a^{r_1}$. It follows from  (\ref{eqn-root32}) that
$$ 
y_2 + \frac{b}{y_2} = \frac{bu}{y_2^2}. 
$$ 
Hence, 
\begin{eqnarray}\label{eqn-tom2}
\frac{b}{y_2^2}=u^{-1}y_2 + \frac{bu^{-1}}{y_2}. 
\end{eqnarray} 

Multiplying $u^{-1}y_2$ to both sides of (\ref{eqn-tom2}), we obtain 
\begin{eqnarray}\label{eqn-tom3}
\frac{bu^{-1}}{y_2}=u^{-2}y_2^2 + bu^{-2}. 
\end{eqnarray} 

Plugging (\ref{eqn-tom3}) into (\ref{eqn-tom2}), we obtain that  
\begin{eqnarray}\label{eqn-root33}
\tr_{m_1}\left(\frac{b}{y_2^2}\right )=\tr_{m_1} \left(\frac{b}{u^2}\right )=
\tr_{m_1} \left(1 + \frac{v}{u^2}\right ). 
\end{eqnarray} 
It then follows from (\ref{eqn-tom}) that 
$$ 
\tr_{m_1}\left(1+\frac{b}{y_2^2}\right )=\tr_{m_1} \left( \frac{v}{u^2}\right )=1. 
$$ 
This is contrary to (\ref{eqn-root36}). Hence we complete the proof. 
\end{proof}

%\end{document} 

We shall use the following lemma in the sequel. The proof of the following lemma employs 
a trick introduced by Dobbertin \cite{Dobb992,Dobb02}.

\begin{lemma}\label{lem-yuan}  
Let $k$ be a positive integer and $q$ be a prime power with $q\not\equiv 0 \pmod{3}$, 
and let $\bar{y}$ denote $y^{q^{m/2}}$, where $m$ is an even positive integer. Then the 
equation
\begin{equation}\label{eqa} 
y^{2k}+y^k\bar{y}^k+\bar{y}^{2k}=0
\end{equation}
has the only solution $y=0$
in $\gf_{q^m}$ if and only if one of the following conditions holds:
\begin{itemize}
\item[(i)] $m \equiv 0 \pmod{4}$;   

\item[(ii)] $q \equiv1 \pmod{3}$;

\item[(iii)] $m \equiv 2 \pmod{4}$, $q\equiv2\pmod{3}$ and $\exponent_3(k)\ge \exponent_3(
q^{m/2}+1)$, where  $\exp_3(i)$ denotes the exponent of $3$ in the cannonical factorization 
of $i$. 
\end{itemize} 
\end{lemma}

\begin{proof} 
Suppose $y\ne0$ and
$y^{2k}+y^k\bar{y}^k+\bar{y}^{2k}=0$. Let $g$ be a primitive element
of $\gf_{q^m}$. Then $\omega=g^{(q^m-1)/3}$ is a primitive third
root of unity in $\gf_{q^m}$. Suppose $y=g^t$ with $1\le t\le
q^m-1$. Then we have
$$\left(\frac{\bar{y}}{y}\right)^{2k}+\left(\frac{\bar{y}}{y}\right)^{k}+1=0.$$
Note that $q\not\equiv 0 \pmod{3}$. We obtain that 
$$\left(\frac{\bar{y}}{y}\right)^{k}=\omega^i, \,\, i \in \{1,2\}.$$
Since $\bar{y}/y=y^{q^{m/2}-1}$ and $y=g^t, \omega=g^{(q^m-1)/3}$,
we have
$$
g^{(q^{m/2}-1)k\cdot t}=g^{i(q^{m}-1)/3}, \ \ i \in \{1,2\}. 
$$ 
It follows that
$$
(q^{m/2}-1)k\cdot t\equiv i(q^{m}-1)/3\pmod{q^{m}-1}, \ \ i \in \{1,2\}. 
$$ 
Since $3|q^m-1$ and
$\gcd((q^{m/2}-1)k, q^{m}-1)=(q^{m/2}-1)\gcd(k, q^{m/2}+1)$, so the
above congruences have no integer solutions $t\in[1, q^m-1]$ if and only if
\begin{equation}\label{eq1}
\exp_3[(q^{m/2}-1)\gcd(k, q^{m/2}+1)]=\exp_3(q^{m}-1).
\end{equation}

In Cases (i) and (ii),  we have that $\exp_3(q^{m/2}-1)=\exp_3(q^{m}-1)$ 
since $\exp_3(q^{m/2}+1)=0$, which
implies (\ref{eq1}) holds. Therefore (\ref{eqa}) has the only
solution $y=0$.

In Case (iii), we have
$\exp_3(q^{m/2}+1)=\exp_3(q^{m}-1)$ since $\exp_3(q^{m/2}-1)=0$. 
Therefore (\ref{eq1}) holds if and only if $\exp_3(k)\ge
\exp_3(q^{m/2}+1)$. Hence (\ref{eqa}) has the only solution $y=0$
if and only if $\exp_3(k) \ge \exp_3(q^{m/2}+1)$. 
\end{proof}

\begin{theorem}\label{thm-yuanpz} 
Let $k$ be a positive integer and $q$ be a prime power with $q \not\equiv 0 \pmod{3}$, and 
let $m$ be an even positive integer. Then
\begin{equation}\label{eq2}
f(x)=x+x^{kq^{m/2}-(k-1)}+x^{(k+1)-kq^{m/2}}\end{equation} is a
permutation polynomial of $\gf_{q^m}$ if and only if one of the
following three conditions holds:
\begin{itemize}
\item[(i)] $m \equiv 0 \pmod{4}$;   

\item[(ii)] $q \equiv1 \pmod{3}$;

\item[(iii)] $m \equiv 2 \pmod{4}$, $q\equiv2\pmod{3}$ and $\exponent_3(k)\ge \exponent_3(
q^{m/2}+1)$, where  $\exp_3(i)$ denotes the exponent of $3$ in the cannonical factorization 
of $i$. 
\end{itemize} 
\end{theorem}

\begin{proof}  
Note that $f(x)=a$ can be written as 
\begin{equation}\label{eq3}
x+\frac{\bar{x}^k}{x^{k-1}}+\frac{x^{k+1}}{\bar{x}^{k}}=a, 
\end{equation}
where we assume $x \ne 0$. So we have
\begin{equation}\label{eq4}
x^{2k}+x^k\bar{x}^k+\bar{x}^{2k}=ax^{k-1}\bar{x}^{k}. 
\end{equation}
Raising both sides of (\ref{eq4}) to the power of $q^{m/2}$, we obtain 
\begin{equation}\label{eq5}
x^{2k}+x^k\bar{x}^k+\bar{x}^{2k}=\bar{a}\bar{x}^{k-1}x^{k}.
\end{equation}
If $a=0$, then by Lemma \ref{lem-yuan}, $f(x)=0$ has the only solution
$x=0$ if and only if one of the three conditions in this theorem holds.

Now we assume that $a\ne0$ (hence $x\ne 0$). Comparing (\ref{eq4})
and (\ref{eq5}), we have $a\bar{x}=\bar{a}x$. It follows that 
\begin{equation}\label{eq6}
\bar{x}=\frac{\bar{a}x}{a}.
\end{equation} 
Using (\ref{eq6}), we
reduce (\ref{eq4}) to
\begin{equation}\label{eq7}
(a^{2k}+a^k\bar{a}^k+\bar{a}^{2k})x^{2k}=a^{k+1}\bar{a}^{k}x^{2k-1}. 
\end{equation}
Since $a\ne0, x\ne0$, by Lemma \ref{lem-yuan}, Equation (\ref{eq7}) has a
unique root for any $a\ne0$ if and only if one of the three conditions
in this theorem holds. This completes the proof. 
\end{proof}

In Theorem \ref{thm-yuanpz} putting $q=2$ and $k=2$, we obtain the following.

\begin{corollary}\label{thm-4first}

Suppose $4\mid m$. Then
$$
f(x)=x+x^{2^{m/2+1}-1}+ x^{2^m-2^{m/2+1}+2}
$$
is a permutation polynomial of $\ftwom$.
\end{corollary}

As a byproduct we have the following.

\begin{corollary}\label{cor-byproduct}
Suppose $4\mid m$. Then the polynomial
$$
g(x)=x^{2^{(m+1)/2}+3}\cdot\left(x^4+x^{2^{(m+3)/2}+2}+x^{2^{(m+5)/2}}\right)^{2^m-2}
$$
is a permutation polynomial of $\ftwom$.
\end{corollary}

\begin{proof} 
Let $f(x)$ be the permutation polynomial defined in Corollary \ref{thm-4first}. 
The proof of Theorem \ref{thm-yuanpz} showed that for any
$x,a\in \ftwom$, if $f(x)=a$, then $x=g(a)$. Thus $g$ is the
compositional inverse of the permutation $f$, and so $g$ also
induces a permutation of $\ftwom$. 
\end{proof} 

In Theorem \ref{thm-yuanpz} putting $q=2$ and $k=1$, we obtain the following.   

\begin{corollary}\label{thm-4second}
Suppose $4\mid m$. Then
$$
f(x)=x+x^{2^{m/2}}+ x^{2^m-2^{m/2}+1}
$$
is a permutation polynomial of $\ftwom$.
\end{corollary}

\begin{corollary}
Suppose $4\mid m$. Then the polynomial
$$
g(x)=x^{2^{(m+1)/2}+2}\cdot\left(x^2+x^{2^{(m+1)/2}+1}+x^{2^{(m+3)/2}}\right)^{2^m-2}
$$
is a permutation polynomial of $\ftwom$.
\end{corollary}

\begin{proof} 
The  proof  is almost identical to that of Corollary \ref{cor-byproduct}, and is omitted here.
\end{proof}

\section{Summary and concluding remarks} 

It looks difficult to characterize permutation binomials over finite fields \cite{Hou}.  To our knowledge, no simple characterization of 
permutation trinomials over finite fields exists in the literature. Hence, it is interesting to construct explicit permutation 
trinomials. Although the objective of this paper is to prove the permutation property of these trinomials, we would also 
mention a few applications of these permutation trinomials.     

The permutation trinomials over $\gf_{2^m}$ presented in this paper can be employed to construct binary linear codes 
within the framework of \cite{CDY}. They can also be used to construct binary cyclic codes with the approach described 
in \cite{Ding131}.  We will treat these applications of these permutation trinomials over $\gf_{2^m}$ in coding theory 
in a piece of future work.   

These permutation trinomials can also be plugged into the Maiorana-McFarland construction 
to obtain bent functions \cite{Dillion74,McFarland73}. The supports of these bent functions are Hadamard difference 
sets whose incidence matrices define binary linear codes \cite{AssmusKey92}. This is another way to construct linear 
codes from these permutation trinomials.

\end{document}